\documentclass[runningheads]{llncs}

\usepackage{tabularx}
\usepackage{algorithm}
\usepackage{multicol}
\usepackage{lipsum}
\usepackage{subcaption}
\usepackage{amsmath}
\usepackage{algpseudocode}
\usepackage{xspace}
\usepackage{multirow}

\usepackage{cite}
\usepackage{subfiles}
\usepackage[hidelinks,colorlinks=true,linkcolor=blue,citecolor=blue,urlcolor=blue]{hyperref}
\usepackage{bookmark}
\usepackage{graphicx}

\usepackage{xcolor}
\usepackage{tikz}
\algdef{S}[FOR]{ForEach}[1]{\algorithmicforeach\ #1\ \algorithmicdo}
\algrenewcommand\algorithmicrequire{\textbf{Input:}}
\algrenewcommand\algorithmicensure{\textbf{Output:}}
\algnewcommand\algorithmicforeach{\textbf{for each}}
\newcommand{\cref}[1]{\S\ref{#1}}

\newcommand{\code}[1]{\texttt{#1}\xspace}

\newcommand*\fname{TimeClave}
\newcommand*\oramname{RoORAM}

\begin{document} 
\title{\fname{}: Oblivious In-enclave Time series Processing System}
\author{}
\institute{}

\author{Kassem Bagher\inst{1,3} \and
Shujie Cui\inst{1} \and
Xingliang Yuan\inst{1} \and
Carsten Rudolph\inst{1} \and
Xun Yi\inst{2}
}
\authorrunning{K. Bagher et al.}

\institute{Monash University,  Australia \and
RMIT University, Australia \and
Faculty of Computing and Information Technology, King Abdulaziz University,  Saudi Arabia}

\maketitle

\begin{abstract}
Cloud platforms are widely adopted by many systems, such as time series processing systems, to store and process massive amounts of sensitive time series data. Unfortunately, several incidents have shown that cloud platforms are vulnerable to internal and external attacks that lead to critical data breaches. Adopting cryptographic protocols such as homomorphic encryption and secure multi-party computation adds high computational and network overhead to query operations.

We present \fname{}, a fully oblivious in-enclave time series processing system: \fname{} leverages Intel SGX to support aggregate statistics on time series with minimal memory consumption inside the enclave. To hide the access pattern inside the enclave, we introduce a non-blocking read-optimised ORAM named \oramname{}. \fname{} integrates \oramname{} to obliviously and securely handle  client queries with high performance. With an aggregation time interval of $10s$, $2^{14}$ summarised data blocks and 8 aggregate functions, \fname{} run point query in $0.03ms$ and a range query of 50 intervals in $0.46ms$. Compared to the ORAM baseline, \fname{} achieves lower query latency by up to $2.5\times$  and up to $2\times$ throughput, with up to 22K queries per second.
\end{abstract}



\section{Introduction}\label{sec:introduction}
Time series data (TSD) are data points collected over repeated intervals, such as minutes, hours, or days. Unlike static data, TSD represents the change in value over time, and analysing it helps understand the cause of a specific pattern or trend over time. 
Studies have been conducted in various fields on TSD to build efficient time series systems, to name a few, healthcare~\cite{li2014physiological, liu2015efficient}, smart home \cite{li2014physiological, aminikhanghahi2018real}, and smart vehicles \cite{gupta2020early}. These systems continuously produce massive amounts of TSD that need to be stored and analysed in a timely manner \cite{netflix2018time, vasisht2017farmbeats, gupta2014bolt}, which cannot be efficiently met by general relational database management systems.

For this, time series databases (TSDB) have been designed and deployed on cloud platforms to provide a high ingest rate and faster insights \cite{andersen2016btrdb}, such as Amazon TimeStream \cite{amazontimestream} and InfluxDB \cite{influxdb}. Unfortunately, adopting plaintext TSDBs on cloud platforms to store and process this massive amount of sensitive TSD can lead to critical data breaches, as several incidents have shown that cloud platforms are vulnerable to internal and external attacks \cite{copper2014healthcare,information2020world,fortune500leak}.

One possible solution to protect TSD in the cloud is to adopt cryptographic protocols such as homomorphic encryption (HE) and secure multi-party computation (MPC) to securely store and process TSD \cite{pappas2014blind,popa2011cryptdb,fuller2017sok,vo2021shielddb,burkhalter2020timecrypt,dauterman2021waldo}. For example, TimeCrypt \cite{burkhalter2020timecrypt} adopts partial HE to provide real-time analytics for TSD, while Waldo \cite{dauterman2021waldo} adopts MPC in a distributed-trust setting. Unfortunately, those solutions have two key limitations. The first is the high computational and network communication cost. As demonstrated in Waldo \cite{dauterman2021waldo}, network communication adds up to $4.4\times$ overhead to query operations. Similarly, HE is orders of magnitude slower than plaintext processing \cite{tetali2013mrcrypt,poddar2018safebricks}. The second limitation is that cryptographic protocols, specifically HE, support limited functionalities such as addition and multiplication on integers \cite{poddar2018safebricks}, where performing complex computations on floating-point numbers adds significant overhead and gradually loses accuracy \cite{viand2021sok}. Furthermore, basic functionalities in time series systems (such as $max$ and $min$) require a secure comparison, which is a costly operation using HE and MPC  \cite{sathya2018review}.

A more practical solution is to adopt hardware-based approaches such as Intel SGX \cite{mckeen2013innovative} to process plaintext TSD in a secure and isolated environment in the cloud, i.e., an enclave. Inside the enclave, data is decrypted and processed in plaintext, allowing systems to securely perform arbitrary and complex computations on the data. However, processing TSD within the enclave is not straightforward due to the costly context switch and access pattern leakage. A context switch occurs when a CPU enters or exists the enclave, such as when the enclave requests encrypted data that reside outside the enclave. Several studies report that context switch is up to $50\times$ more expensive than a system call \cite{weichbrodt2018sgx,tian2018switchless}. 
Although such overhead can be minimised by adopting Switchless Calls techniques \cite{tian2018switchless}, the enclave has to validate the requested data to prevent the untrusted part from tampering with the query. In addition, the enclave needs to decrypt the requested data, which adds a costly overhead, as seen in previous work \cite{sun2021building}.

SGX-based solutions are also vulnerable to access pattern attacks, where an attacker performs a page table-based attack \cite{xu2015controlled,van2017telling,wang2017leaky,moghimi2020copycat} to observe which memory page is accessed inside the enclave. This leakage allows the attacker to infer which data are being accessed and when \cite{cash2015leakage,islam2012access,zhang2016all,van2017telling,xu2015controlled}. Several studies have shown that such information can recover search queries or a portion of encrypted records \cite{giraud2017practical,zhang2016all,liu2014search}.
A widely adopted approach to hide access patterns is to store encrypted data in an Oblivious RAM (ORAM).
Several solutions combine Intel SGX with ORAM to hide access patterns inside the enclave.
For example, Oblix \cite{mishra2018oblix} and ZeroTrace \cite{sasy2017zerotrace} deploy the ORAM controller securely inside the enclave while leaving the ORAM tree outside the enclave, which increases the communication between the enclave and the untrusted part. Such communication adds additional overhead to the system due to context switching and additional processing performed by the enclave (e.g., data decryption and validation), which degrades the system's performance.
Even when deploying the ORAM inside the enclave, these solutions adopt vanilla ORAMs, which are not optimised to handle non-blocking clients' queries that dominate the workload of read-heavy systems like TSDB. \\

Motivated by the above challenges, we ask the following question: \textit{How can we build an oblivious SGX-based system capable of securely storing TSD and answering clients' queries with high performance while hiding access patterns?}
\\

\noindent\textbf{\fname{}.} 
To answer this question, this paper presents \fname{}, an oblivious in-enclave time series processing system that efficiently stores and processes TSD inside the enclave. \fname{} resides entirely inside the enclave and adopts oblivious primitives to provide a fully oblivious in-enclave time series processing system. \fname{} supports oblivious statistical point and range queries by employing a wide range of aggregate and non-aggregate functions on TSD that are widely adopted in the area of time series processing \cite{amazontimestream,influxdb, timescale}, such as $sum$, $max$, and $stdv$. 

To efficiently protect against access pattern leakage inside the enclave, we introduce Read-optimised ORAM (\oramname{}), a non-blocking in-enclave ORAM. As time series systems need to store and query TSD simultaneously, most ORAMs fail to meet such requirements. The reason is that ORAMs can perform only one operation at a time, i.e. read or write. Regardless of the required operation, ORAM reads and writes data back to the ORAM tree to hide the operation type. As a result, clients' queries are blocked during a write operation and delayed during read operations, especially in read-heavy systems. An efficient approach to address the previous drawbacks is to decouple read and write operations to handle non-blocking client queries. Existing solutions follow different approaches to support non-blocking read operations and parallel access to the ORAM \cite{sahin2016taostore,chakraborti2018concuroram}. However, these solutions require either a proxy server to synchronise the clients with the cloud server \cite{sahin2016taostore} or require the client to maintain a locally-cached sub-tree that is continuously synchronised with the cloud server \cite{chakraborti2018concuroram}. Such a client or third-party server involvement degrades the system's performance and usability, especially in read-heavy systems like time series processing, which requires real-time and low-latency query processing.
\oramname{} adopts and improves PathORAM \cite{stefanov2018path} to handle client queries with better performance.  \oramname{} achieves this improvement by having separate ORAM trees, a read-only and a write-only tree, to decouple read- and write operations. Unlike previous solutions \cite{chakraborti2018concuroram,sahin2016taostore}, \oramname{} is an in-enclave and lightweight; hence, it does not require client involvement or a proxy server. \oramname{} adopts oblivious primitives inside the enclave to access the ORAM controller obliviously. 

To avoid the costly context switches, \fname{}  integrates \oramname{} to efficiently and obliviously store and access TSD inside the enclave, thus, eliminating excessive communication with the untrusted part. 
In detail, \fname{} leverages the fact that TSD is queried and aggregated in an approximate manner to store statistical summaries of pre-defined time intervals inside the enclave. 
By only storing summarised TSD, \fname{} reduces the size of the ORAM tree, reducing the consumption of enclave memory and the cost of ORAM access.
Such a design allows \fname{} to store and process a large amount of TSD data with low memory consumption while providing low-latency queries.

\noindent\textbf{Performance evaluation}
We implemented and evaluated \fname{} on an SGX-enabled machine with 8 cores and 128 GiB RAM \cref{sec:evaluation}. 
With different query ranges of 8 aggregate functions, \fname{} runs a point query in $0.03ms$ and a range query of 50 intervals in $0.46ms$. \fname{} achieves lower query latency of up to $2.5\times$ and up to $2\times$ higher throughout. Finally, \fname{} achieves up to $17 \times$ speed-up when inserting new data compared to the ORAM baseline.
%


\section{Background}\label{sec:background}

\subsection{Intel SGX}\label{sec:intel_sgx}
Intel SGX is a set of commands that allow creating an isolated execution environment, called an \code{enclave}. Applications run their critical code and store sensitive data inside the enclave. The enclave is designed to protect sensitive code and data from a malicious host. 
No other process on the same CPU (except the one that created the enclave), even the privileged one (kernel or hypervisor), can access or tamper with the enclave.
SGX-based applications are divided into two parts: untrusted and trusted parts. The two parts communicate with each other using user-defined interface functions, i.e., \code{ECALL} and \code{OCALL}. The \code{ECALL} allows the untrusted part to invoke code inside the trusted part (enclave) and send encrypted data to the enclave. \code{OCALL} allows the enclave to send encrypted data to the untrusted part.
In addition, SGX provides a remote attestation feature, allowing the client to authenticate the enclave's identity, verify the integrity of the enclave's code and data, and share encryption keys. We refer the readers to \cite{costan2016intel} for more details about Intel SGX. During remote attestation, the client and the enclave establish a secure channel for secure communication and share their encryption keys.

\subsection{ORAM}
Oblivious RAM (ORAM) was introduced by Goldreich and Ostrovsky \cite{goldreich1996software} to protect client data on an untrusted remote machine (such as the cloud) against access pattern attacks. The main idea of ORAM is to hide the user's access patterns by hiding the accessed memory addresses on the cloud, hence, making the user's access oblivious. To achieve this, ORAM continuously shuffles and re-encrypts a portion or all of the user's accessed data in a trusted memory region or on a trusted third-party machine. Since then, many ORAM schemes have been proposed, to list a few, \cite{boneh2011remote, damgaard2011perfectly,gentry2013optimizing,goldreich1996software,goodrich2011privacy,ostrovsky1990efficient}. As of writing this paper, the most notable among them is the widely adopted PathORAM \cite{stefanov2018path}.

PathORAM uses a complete binary tree of height $L=\left[\log _{2} N\right\rceil-1$ to store $N$ encrypted records on untrusted storage, e.g., an untrusted cloud. Every node in the tree is a bucket, where each bucket contains a fixed number ($Z$) of encrypted data blocks. The tree contains real blocks (client's blocks) and dummy blocks. Each data block is randomly assigned to a leaf between 0 and $2^L -1$. The client maintains a \code{Position Map} that tracks the leaf to which each block points in the tree. To access a block in the tree, the client retrieves all blocks on the path from the root node to the leaf node of the required block and stores them in the stash. The client then assigns a random leaf to the accessed/updated block. All blocks are then re-encrypted and inserted in buckets where empty buckets are filled with dummy blocks. The accessed path is then written back to the tree (cloud). As a result, the cloud cannot tell which block is accessed by the client or identify the blocks.
In  \oramname{}, we adopt and enhance PathORAM to support non-blocking queries while achieving high performance with a bit of relaxation on the security guarantee (See \oramname{} \cref{sec:oram_security_guarantees}).

\subsection{Oblivious Primitives}\label{sec:oblivious_primitives}
Oblivious primitives are used to access or move data in an oblivious manner, i.e., without revealing access patterns. For example, accessing the $i$-th item in an array can reveal the memory address of the item. However, an oblivious access primitive accesses the item without revealing its address. Therefore, hiding the memory access patterns. We use a set of data-oblivious primitives based on previous works \cite{ohrimenko2016oblivious,poddar2020visor}.
We use the following primitives from a library of general-purpose oblivious primitives provided in \cite{law2020secure}:

\noindent -- \textbf{Oblivious comparisons}. The following primitives are used to obliviously compare two variables; \code{oless(x,y)}, \code{ogreater(x,y)} and \code{oequal(x,y)}.

\noindent -- \textbf{Oblivious assignment}. The \code{oassign(cond,x,y)} is a conditional assignment primitive, where a value is moved from a source to a destination variable if the condition is true. Typically, the \code{oassign} is used with the oblivious comparisons primitives to compare values and return a result obliviously.

\noindent -- \textbf{Oblivious array access}. The \code{oaccess(op, arr, index)} primitive is used to read/write an item in an array without revealing the item's address.

\noindent -- \textbf{Oblivious exists}. The \code{oexists(arr,x)} primitive is used to obliviously determine if a giving item exists in a giving array or not. \code{oexists} achieves this by combining \code{oaccess} and \code{oequal}.


\section{System Overview}\label{sec:system_overview}

\begin{figure}
\vspace{-25pt}
\centering
    \includegraphics[width=0.85\linewidth]{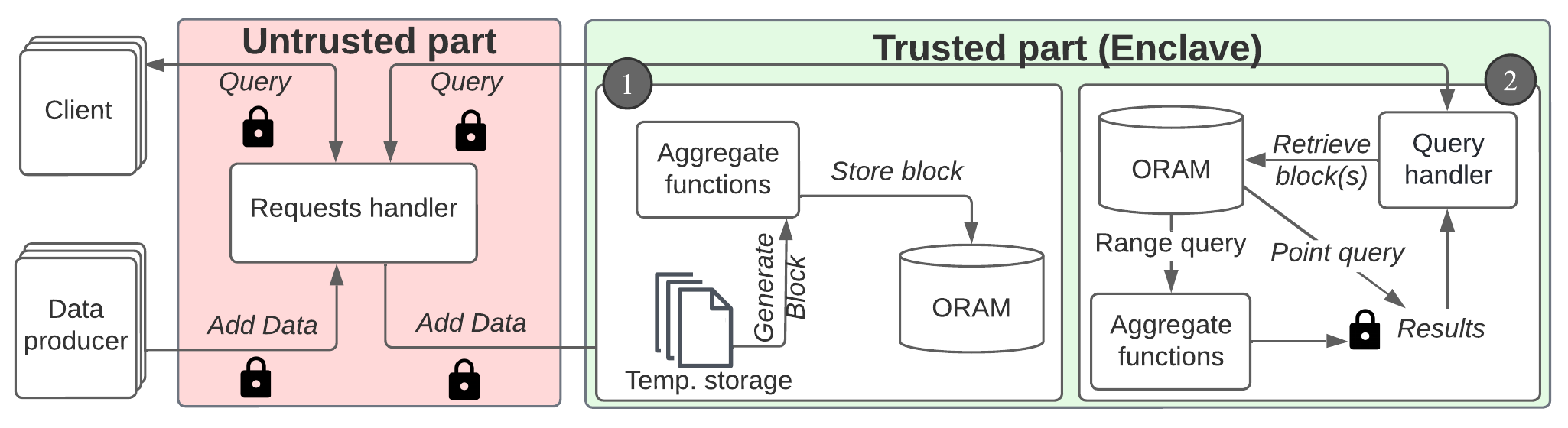}
    \caption{\fname{} architecture. The figure illustrates how \fname{} securely and obliviously stores time series data (1), and how it handles clients' queries (2).}
    \label{fig:architecture}
    \vspace{-16pt}
\end{figure}

\fname{} consists of 4 entities: data producers, clients, server (referred to as the cloud or server), and the SGX enclave running on the server. Data producers, such as sensors or other devices, generate raw time series data and upload them to the server. Clients send encrypted queries to the server. Finally, the server (typically deployed in the cloud) stores the time series data and handles clients' queries.

\subsection{Threat Model}\label{sec:threat_model}
We consider a semi-trusted server where an adversary is interested in learning clients' data and queries. Furthermore, we consider an adversary who has full control over the server except for the CPU (Intel SGX). Therefore, the adversary can obtain the encrypted client queries and data, but cannot examine them in plaintext. In addition, the adversary can see the entire memory trace, including the enclave memory (at the page level). We consider trusted data producers and clients; therefore, only authorised users can submit queries to the server. We do not consider DoS attacks or other side-channels on the enclave, e.g., speculative execution \cite{BulckForeshadow18}, voltage changes \cite{chen2021voltpillager}, or cache attack \cite{gullasch2011cache}. We discuss the of security of \oramname{} in \cref{sec:oram_security_guarantees} and \fname{} in \cref{sec:framework_security_guarantees}.

\subsection{Overview of \fname{}}

The architecture of \fname{} is depicted in Fig.~\ref{fig:architecture}. We can see that \fname{} is partitioned into two parts on the server: untrusted and trusted. The untrusted part runs on the host outside Intel SGX while the trusted part runs inside Intel SGX (referred to as the enclave or enclave code). The role of the untrusted part is to facilitate communication between the client (including the data producer) and the enclave. This is achieved by receiving encrypted requests from the client and sending them to the enclave. 

To protect data and queries from the cloud, data producers, clients, and the enclave share a secret key $k$. 
The raw time series data and queries are encrypted with $k$ before being sent to the cloud. Note that clients are authenticated and $k$ is shared during SGX Remote Attestation (RA). Also, clients and data producers authenticate the enclave's identity and verify the code's integrity during RA.  

As depicted in Figure \ref{fig:architecture}, most of \fname{} components run inside the enclave. \fname{} stores the data in an ORAM inside the enclave, i.e., \oramname{} to protect against SGX memory access patterns. The reason of deploying \fname{} entirely inside the enclave is that new generations of SGX support up to 1 TB of enclave memory, allowing applications to store more data and run larger applications inside the enclave. \fname{} utilises the enclave’s large memory to securely store and process data securely within the enclave \cite{intelsgx1tb}.

The trusted and untrusted components of \fname{} cooperate to perform two main functionalities:
storing time series data and processing client queries.

\noindent\textbf{1) Storing time series data}. The data producer generates and encrypts raw time series data and transmits them to the server. A request handler on the untrusted part forwards it to the enclave. The enclave then decrypts the data and stores them temporarily for a pre-defined time interval ($T$). When $T$ elapses, \fname{} (inside the enclave) calculates the aggregated values (using all supported aggregate functions) for the data points cached in temporary storage and stores the values in a block. Once the block is generated, the raw time series data are discarded, and the block is stored in the ORAM, i.e., \oramname{}.

\noindent\textbf{2) Clients’ Query}. The client first generates a query $\mathcal{Q}= \langle f,(t_a,t_b) \rangle$, where $f$ is the aggregate function, $t_a$ and $t_b$ are the time range. The client then encrypts the query, i.e., $\mathcal{Q'} = \mathcal{E}_k(\mathcal{Q})$, where $\mathcal{E}$ is a symmetric encryption scheme. The encrypted query is then transmitted to the server. Once the request handler receives the query, it sends it to the enclave. The enclave first decrypts the query by computing $\mathcal{Q} = \mathcal{D}_k(\mathcal{Q'})$, and then sends it to the query optimiser. The query optimiser is responsible for optimising the number of accessed blocks in the ORAM tree by efficiently accessing and merging blocks from different aggregation intervals (For more details refer to \ref{sec:query_realisation}. The results are encrypted using the client’s key and returned to the client through a secure channel established through RA.

\noindent \textbf{Supported Aggregate Functions}:\label{sec:supported_aggregates}
\fname{} supports a set of additives aggregates, i.e., sum, count, mean, variance and standard deviation. In addition, \fname{} supports a set of more complex non-additive aggregates, i.e., max and min. \fname{} uses the previous functions on a set of data points to generate summarised blocks for a pre-defined time interval (see \cref{sec:data_model}).  These functions are used to answer simple and complex queries, where multiple aggregated values are used in the calculations instead of raw data points (see \cref{sec:query_realisation}). 


\section{\oramname{}}\label{sec:proposed_oram}
In this section, we introduce our proposed ORAM, namely \oramname{}, which is integrated into \fname{} to provide oblivious data storage inside the enclave capable of handling non-blocking read-operations (i.e., clients' queries). Later, we describe how \fname{} efficiently stores time series data in \oramname{} and how it is used to realise clients’ queries.

\subsection{Overview}\label{sec:oram_overview}
The main idea of \oramname{} is to decouple the eviction process from the read/write operation. This separation allows \oramname{} to evict the accessed paths without blocking clients' queries. Inspired by \cite{chakraborti2018concuroram}, \oramname{} performs read-operations (queries) on a read-only tree and write-operations (writing data and eviction) on a write-only tree. However, unlike \cite{chakraborti2018concuroram}, \oramname{} stores and obliviously accesses the controller and components inside the enclave (as illustrated in Figure \ref{fig:oram_overview}). Further, \oramname{} does not require clients' involvement to maintain a locally-cached sub-tree nor requires synchronisation with the server. Such a lightweight design makes \oramname{} tailored for in-enclave read-heavy time series processing systems, which require real-time and low-latency query processing.

\oramname{} uses the dedicated read-only tree to perform query operations only (read operation). This allows \oramname{} to perform multiple read-operations on the read-only tree (reading multiple paths) before evicting the accessed paths. The retrieved blocks are stored in the stash. After $\mathcal{R}$ read-operations, \oramname{} evicts the blocks from the stash to the write-only tree and then synchronises both trees. \oramname{} notations are defined in Table~\ref{tbl:notations}.
\begin{figure}
\centering
    \includegraphics[width=0.8\linewidth]{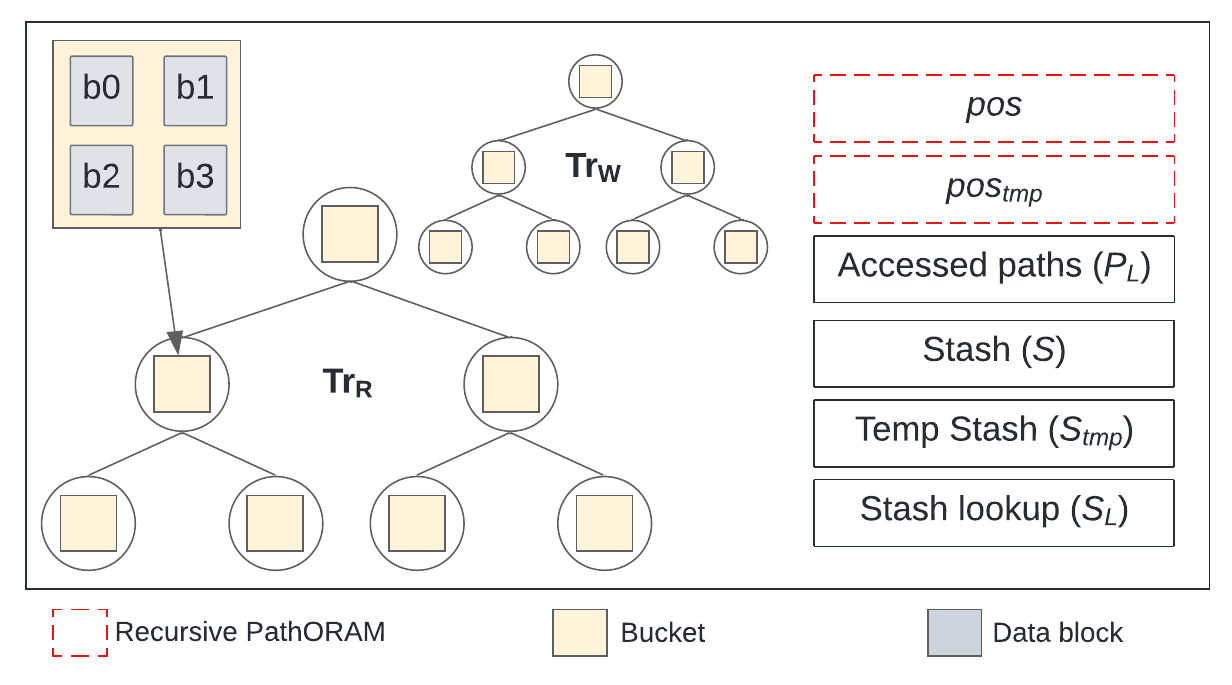}
    \caption{\oramname{} structure and components. The position map ($pos$) and the temporary position map ($pos_{tmp}$) are stored in recursive PathORAM. $S$, $S_{tmp}$, $S_L$ and $\mathcal{P}_L$ are stored in an array and are accessed using oblivious primitives \cref{sec:oblivious_primitives}.}
    \label{fig:oram_overview}
\end{figure}

\subsection{Structure and Components}

\begin{table}
    \vspace{-20pt}
    \caption{\oramname{} notations. $^{*}$ Represents notations introduced by \oramname{}.}
    \label{tbl:notations}
    \vspace{0.2cm}
    \renewcommand{\arraystretch}{1.2}
    \begin{tabular}{|p{0.18\linewidth} | p{0.8\linewidth}|}
        \hline
        \textbf{Notation} & \textbf{Meaning} \\
        \hline$N$ & Total number of blocks on the server \\
        \hline $L$ & Height of binary tree. $L=\left[\log _{2} N\right\rceil-1$). \\
        \hline $B$ & Block size in bytes \\
        \hline $Z$ & Bucket capacity (in blocks) \\
        \hline $\mathcal{P}(x)$ & Path from leaf node $x$ to the root \\
        \hline $\mathcal{P}(x, \ell)$ & The bucket at level $\ell$ along the path
        $\mathcal{P}(x)$ \\
        \hline $\mathcal{P_L}$$^{*}$ & Path lookup, a list of accessed paths' IDs by the read operation \\
        \hline $\mathcal{R}$$^{*}$ & Eviction frequency, number of read operations before batch eviction \\
        \hline $S$ & Read stash \\
        \hline $S_{tmp}$$^{*}$ & Write Stash, a temporary stash used during eviction. \\
        \hline $S_{L}$$^{*}$ & Stash blocks lookup \\
        \hline $pos$ & Position map used for read operations \\
        \hline $pos_{tmp}$$^{*}$ & Temporary position map used during batch eviction. \\
        \hline $Tr_R$ & Read-only tree, used for read operations \\
        \hline $Tr_W$$^{*}$ & Write-only tree, used during batch eviction\\
        \hline
    \end{tabular}
    \vspace{-10pt}
\end{table}

\noindent\textbf{Binary tree}. Similar to PathORAM, \oramname{} stores data in a binary tree data structure of height $L=\left\lceil\log _{2}(N)\right\rceil-1$ and $2^{L}$ leafs.

\noindent\textbf{Bucket}. Each node in the tree contains a bucket, where each bucket contains $Z$ real data blocks. Buckets containing less than $Z$ blocks are filled with dummy data.

\noindent\textbf{Path}. A path represents a set of buckets from the leaf node $x$ to the root node. $\mathcal{P}(x)$ denotes the path from leaf node $x$ to the root and $\mathcal{P}(x, \ell)$ denotes the bucket at level $\ell$ along the path $\mathcal{P}(x)$.

\noindent\textbf{Block}. A block contains summarised data for a specific pre-defined time interval. Each block is assigned a random path in the tree between 0 and $2^{L} -1$. Accessing a block is achieved by accessing a path $\mathcal{P}(x)$ on the read-only tree $Tr_R$. 

\noindent\textbf{Stash $S$}. When a path is accessed, blocks are stored and kept in the stash $S$ until batch eviction. During batch eviction, the items in the stash $S$ are moved to a temporary stash $S_{tmp}$, allowing query operations to insert blocks into $S$. Stash and temporary stash has a size of $O(\log_{2} N) \cdot \mathcal{R}$.
Notice that the stash avoids block duplication by storing unique blocks only while replacing duplicated blocks with dummy data to avoid information leakage. 

\noindent\textbf{Stash lookup $S_L$}. Stash lookup $S_L$ contains only the IDs of the retrieved blocks and is used to answer whether a block is in the stash or not. The cost of accessing $S_L$ is lower than $S$ as $S_L$ contains smaller-sized data than $S$. Similar to $S$, $S_L$ has a worst case size of $O(\log_{2} N) \cdot \mathcal{R}$.

\noindent\textbf{Position map $pos$}. The position map stores the path to which each block belongs. The position map is updated every time a block is accessed. \oramname{} stores the position map in a recursive PathORAM \cite{stefanov2018path} instead of an array to achieve obliviousness. The reason is that the position map contains large number of items, therefore, storing these items and accessing them linearly has a high cost compared to a recursive PathORAM. 

\noindent\textbf{Path lookup $\mathcal{P}_L$}. It stores the list of accessed paths $\mathcal{P}_L$ (leaf nodes' IDs) that have been accessed during read-operations (query). $\mathcal{P}_L$ is used during a batch eviction to write the accessed paths back to the tree.
\oramname{} clears the list after each batch eviction; thus, the maximum size of the list is $\mathcal{R}$.

\subsection{Initialisation}
Both read- and write-trees are initialised with height $L=\left[\log _{2} N\right\rceil-1$. 
Therefore, each tree contains $2^{L+1} -1$ buckets, where each bucket is filled with dummy blocks.
Position maps are initialised with an independent random number between 0 and $2^{L}-1$. 
Stash, temporary stash, and stash lookup are initialised with empty data.  
Path lookup $\mathcal{P}_L$ is initialised with empty data with size $\mathcal{R}$.

\subsection{Read Operation}\label{sec:read_access}

\begin{algorithm}

\scriptsize
\caption{Read Operation}\label{alg_read_access}
\begin{multicols}{2}
\begin{algorithmic}[1]

\Require $b_{id}$ - Block id
\Ensure Summarised data block
\Function{ReadAccess}{$b_{id}$}
\State $x \gets pos[b_{id}]$ 
\State $pos_{tmp}[b_{id}] \gets$ \\
    \hspace{1cm} $UniformRandom(0...2^{L}-1)$
\color{purple}
\If{ \texttt{oexists}$(b_{id}, S_{L}$\texttt{)}}
    \State  $ReadPath(Tr_{R}, dummy)$
    \State \texttt{oaccess}$(write,\mathcal{P}_L, dummy$\texttt{)}
\Else
    \State $ReadPath(Tr_{R}, x)$
    \State \texttt{oaccess}$(write,\mathcal{P}_L,x$\texttt{)}
\EndIf
\color{black}
\State \texttt{oassign}$(true, d',$ \texttt{oaccess(}$read,S, b_{id}$\texttt{))}
\State \Return $d'$
\EndFunction

\end{algorithmic}
\hfill
\begin{algorithmic}[1]
\item[]
\Function{ReadPath}{$ORAM, x$}
    \For{$l \in$ \{ $L,L-1...0$ \}}         
        \State \texttt{oaccess}$(write,S,$ $GetBucket(P(x,l))$
        \color{purple}
        \State \texttt{oaccess}$(write,S_L,$ \\
        \hspace{1cm} $GetBucket(P(x,l)).b_{id}$
        \color{black}
    \EndFor
\EndFunction
\end{algorithmic}
\end{multicols}
\end{algorithm}

The details of \code{READACCESS} are shown in Algorithm~\ref{alg_read_access}. It is worth noting that, aside from the distinct design variations between \fname{} and PathORAM, the algorithmic distinctions are also demonstrated in Algorithm \ref{alg_read_access}, \ref{alg_write_access}, and \ref{alg_eviction}, which are highlighted in red.
To access block $a$, given its block ID $b_{id}$, \oramname{} first accesses the position map to retrieve the block's position in $Tr_R$, such that $x:= pos[b_{id}]$. 
Second,  a new random path is assigned to block $a$ and stored in the temporary position map ($pos_{tmp}$). \oramname{} updates $pos_{tmp}$ instead of $pos$ as the accessed block will not be evicted before $\mathcal{R}$ read operations. Thus, avoiding inconsistent block position for subsequent queries before a batch eviction.

The next step is to check whether block $a$ is stored in the stash or not by searching $S_L$ with the oblivious primitive \code{oexists} \cref{sec:oblivious_primitives}. 
If $S_L$ contains $a$, a dummy path will be accessed; otherwise, path $x$ is accessed. 
By doing so, the adversary cannot infer whether block $a$ is located in the stash ($S$) or in the ORAM tree ($Tr_R$). 
In both cases, the retrieved blocks of the accessed path are stored in the stash $S$, and its path ID is tracked in $\mathcal{P}_L$.
Finally, block $a$ is obliviously retrieved from the stash $S$ with \code{oaccess} and assigned to $d'$ with \code{oassign}.

\subsection{Write Operation}
\label{sec:oram_write_access}
\begin{algorithm}[t]
\scriptsize
\caption{Write Operation}\label{alg_write_access}
\begin{multicols}{2}
\begin{algorithmic}[1]
\Require $data*$ - Block data, $time$ - block time interval
\Function{WriteAccess}{$data*, time$}
\color{purple}
\State $\underline{QueryLock.lock}$ 
\State $x \gets UniformRandom(0...2^{L}-1)$
\State $b_{id} \gets time$
\State $pos_{tmp}[b_{id}] \gets x$
\State \texttt{oaccess(}$write,S, data*$\texttt{)}
\State \texttt{oaccess(}$write,S_L, b_{id}$\texttt{)}
\State $\underline{QueryLock.unlock}$ 
\color{black}
\EndFunction
\end{algorithmic}
\end{multicols}
\end{algorithm}


The details of \code{WRITEACCESS} is given in Algorithm~\ref{alg_write_access}. 
The data block to be written $data*$ is associated with a time interval $time$, and will be used as the ID of $data*$. 
Since the stash is accessed by both read and write operations, adding a block to the stash requires synchronisation using a mutex (i.e., query lock). Therefore, queries are blocked during a write operation. However, a write operation requires few operations only, such as adding the block to stash and updating the position map, which adds a negligible overhead. Note that if the stash is full, \oramname{} will automatically evict the blocks (see \cref{sec:eviction}).

To write $data*$ to the tree, \oramname{} assigns a random path $x$ to it by setting $pos[time] \leftarrow x$ and obliviously adds the block to the stash $S$ with \code{Oaccess}. Meanwhile, the block ID $time$ is added to $S_L$, as $data*$ is stored in the stash. 
Unlike other tree-based ORAMs, stash items in \oramname{} are not evicted after a write operation. Instead, stash items are evicted in batches after $\mathcal{R}$  read operations \cref{sec:eviction}.

\subsection{Batch Eviction and Trees Synchronisation}
\label{sec:eviction}
\begin{algorithm}[t]
\scriptsize

\caption{Batch Eviction}\label{alg_eviction}
\begin{multicols}{2}
\begin{algorithmic}[1]
\Function{Evict()}{}
    \color{purple}
    \State let $eBuckets$ be the IDs of evicted buckets
    \State $\underline{QueryLock.lock}$ 
    \State $swap(S,S_{tmp})$
    \State $S_L.clear()$
    \State $\underline{QueryLock.unlock}$ 
    \color{black}
    
    \ForEach {$p \in \mathcal{P}_L$}
        \For{$l \in$ \{ $L,L-1...0$ \}}                    
            \color{purple}
            \If{\texttt{!oexists(}$\mathcal{P}(p,l).id, eBuckets$\texttt{)}}
                \color{black}
                \State $ S' \gets (a',data') \in S_{tmp}  : \mathcal{P} (pos_{tmp}[a'],l)$
                \State $S' \gets$ Select $min(|S'|,Z)$ blocks from $S'$
                \color{purple}
                \State $S_{tmp} \gets S_{tmp} - S'$ 
            \State $eBuckets = eBuckets \cup \mathcal{P}(p,L).id$
            \color{black}
            \EndIf
        \EndFor
    \EndFor
    \color{purple}
    \State eBuckets.clear()
    \State $\underline{QueryLock.lock}$ 
    \State Copy changes from $Tr_W$ to $Tr_R$
    \State Copy changes from $pos_{tmp}$ to $pos$
    \State $S \gets S \cup S_{tmp}$
    \State Clear $S_{tmp}$
    \State Clear $\mathcal{P}_L$
    \State $\underline{QueryLock.unlock}$ 
    \color{black}
\EndFunction
\end{algorithmic}
\end{multicols}
\end{algorithm}
\oramname{} performs $\mathcal{R}$ read operations on the read-only tree prior to a batch eviction. As a result, \oramname{} needs to write multiple paths at once in a single non-blocking batch eviction. It is known that eviction in PathORAM is an expensive process; hence, it can degrade the query performance. \oramname{} addresses this issue by blocking only queries during the execution of critical sections in batch evictions. Note that there are a few steps during eviction where \oramname{} needs to block queries. However, these steps have a negligible impact on query performance, making the batch eviction a non-blocking process. 
As shown in Algorithm \ref{alg_eviction}, \oramname{} splits the batch eviction into two phases:

\noindent\textbf{Path Writing Phase (lines 3-17, Algorithm~\ref{alg_eviction})}. During the eviction phase, \oramname{} starts by acquiring a mutex for a short period to swap $S$ and $S_{tmp}$.  Swapping stash items allows query operations (read operations on the ORAM) to insert blocks into $S$ while batch eviction is in process. 
Note that moving stash items is achieved by a simple reference swap instead of swapping data. 
The eviction process writes all the accessed paths recorded in $\mathcal{P}_L$ to the write-only tree (lines 7 to 17).  Specifically, for each path $p$ in $\mathcal{P}_L$, \oramname{} greedily fills the path's buckets with blocks from $S_{tmp}$ in the order from leaf to root. This order ensures that the blocks are pushed into $Tr_W$ as deep as possible. All non-evicted blocks remain in $S_{tmp}$ to be evicted in subsequent batch evictions.

When \oramname{} writes multiple paths to $Tr_W$, there can be an intersection between two paths, at least at the root level.  A bucket may be written several times during a batch eviction (e.g., the root node's bucket), causing a buckets collision. 
\oramname{} avoids that by writing every bucket only once. Such an approach can improve performance by reducing the number of evicted buckets. However, it leaks the number of intersected buckets to the adversary. \oramname{} prevents such leakage by performing fake access to all buckets in the intersected paths.

\noindent\textbf{Synchronisation Phase (lines 18-24, Algorithm~\ref{alg_eviction})}. At this point, $Tr_R$ needs to be synchronised with $Tr_W$ to reflect the new changes. To synchronise the two trees with minimal query blocking (Lines 18 to 24), \oramname{} copies only the written changes (i.e., paths) from $Tr_W$ to $Tr_R$ instead of copying the entire tree. In addition, \oramname{} copies the changes from $pos_{tmp}$ to $pos$ and any non-evicted blocks in $S_{tmp}$ to $S$. 

\subsection{Memory Consumption}\label{sec:storage}
\oramname{} uses two separate ORAM trees to handle non-blocking queries. As a result, it consumes twice the data as the ORAM baseline. In detail, \fname{} maintains two separate ORAM trees to store the data blocks, two stashes, two position maps and the list of accessed paths.
Each tree will have a height $L=\left[\log _{2} N\right\rceil-1$. Therefore, each tree will be able to store up to $2^{L+1} -1$ buckets. Since each bucket contains $Z$ blocks, the number of blocks stored in the tree is $2^{L+1} -1 \cdot Z$. Therefore, each tree will require $2^{L+1} -1 \cdot Z \cdot B$ bytes. Stash ($S$) and Stash Lookup ($S_L$) require $O(\log_{2} N) \cdot \mathcal{R} \cdot B$ bytes. Path Lookup ($\mathcal{P_L}$) requires $\mathcal{R} \cdot 4$ bytes (since each path requires 4-bytes).

\subsection{\oramname{}{} Efficiency Analysis}
\oramname{}'s operation overheads involve four main operations: accessing and updating the temporary position map $pos_{tmp}$, stash lookup table access $S_L$, read-only tree $Tr_R$ path reading, and path lookup access $\mathcal{P}_L$. Each of these operations is associated with a computational cost. For read operations \cref{sec:read_access}, accessing and updating $pos_{tmp}$ and the path reading from $Tr_R$ both involve an asymptotic cost of $O(\log N)$ due to recursive PathORAM and its position map inside the enclave. Accessing $S_L$ costs $O(\log N) \cdot \mathcal{R}$, while accessing $\mathcal{P}_L$ costs $O(\mathcal{R})$. Therefore, the overall cost for a read operation is $O(log N)$, yielding similar asymptotic complexity to PathORAM. Nevertheless, \oramname{} shows higher performance up to 2.5 times \cref{sec:evaluation}. This enhancement is due to \oramname{}'s design of decoupling the non-blocking read operations from the non-blocking eviction process.The write operation \cref{sec:oram_write_access}  accesses $pos_{tmp}$, $S$, and $S_L$. I.e., $O(\log N)$  +  $O(\log N) \cdot \mathcal{R}$  + $O(\log N) \cdot \mathcal{R}$ where $\mathcal{R}$ is a constant. Consequently, the overall cost is $O(log N)$. Finally, path writing during batch eviction requires $O(\log N)$, and tree synchronization, which involves copying $O(\log N) · R$ items, leads to an overall cost of $O(\log N)$.
\subsection{Security of \oramname{}}\label{sec:oram_security_guarantees}
To prove the security of \oramname{}, we adopt the standard security definition for ORAMs from \cite{stefanov2011towards}. An ORAM is said to be secure if, for any two data request sequences, their access patterns are computationally indistinguishable by anyone but the client. 
\oramname{} is similar to PathORAM but excludes two main points: 1) \oramname{} stores all components in the enclave, whereas PathORAM stores the stash and position map in the client; 2) PathORAM evicts the stash data after each access, while \oramname{} performs batch eviction after $R$ read operations. 
Therefore, the security definition of \oramname{} is captured in the following theorem. 

\begin{theorem}\label{def:oram}
\textit{\oramname{} is said to be secure if, for any two data request sequences $\vec{y_1}$ and $\vec{y_2}$ of the same length, their access patterns $A(\vec{y_1})$ and $A(\vec{y_2})$ are computationally indistinguishable by anyone but the enclave.
}\end{theorem}

\begin{proof} To prove the security of \oramname{}, we now provide the following analysis.

Storing and accessing the stash, position map, and other components of \oramname{} within the enclave do not leak additional information to the adversary because they are all accessed with oblivious primitives. 
Based on the security of oblivious primitives, the adversary cannot infer which item is accessed in the position map $pos$, the temporary position map $pos_{tmp}$, the stash $S$, temporary stash $S_{tmp}$, stash lookup $S_L$, and path lookup $\mathcal{P}_L$.

Evicting the stash after $R$ reads can improve the performance of \oramname{} with a bit of relaxation on the security guarantee, compared to the security of Path ORAM.
Note that the security of \oramname{} can reach the standard security via extra dummy access, which will be explained later.
Let $\vec{y}=(a_{R-1}, ..., a_1)$ be a sequence of read-access of size $\mathcal{R}$. The adversary sees a sequence
$A(\vec{y})=\left(\right.$ $pos_{\mathcal{R}}\left[\mathrm{a}_{\mathcal{R}}\right]$, $pos_{\mathcal{R}-1}\left[\mathrm{a}_{\mathcal{R}-1}\right], \ldots$, $pos\left._{1}\left[\mathrm{a}_{1}\right]\right)$. 
It is possible for the adversary to distinguish some read-accesses based on the access pattern, e.g., $\vec{y_1}=(a_1, a_1, ..., a_1)$ and $\vec{y_2}=(a_1, a_2, ..., a_R)$ (here we assume that $a_1$, $a_2$, ..., and $a_R$ are stored in different leaf nodes). 
Recall that within $\mathcal{R}$ reads, if the required block is stored in the stash, a random path will be accessed; otherwise, the required path will be accessed. 
Assume that the tree contains $M = 2^L$ paths. 
In the example, $pos[a_1]$, $pos[a_2]$, ..., and $pos[a_R]$ must be different. For $\vec{y_1}$, any two paths are the same with a probability of $1/M$. However, as long as $A(\vec{y})$ contains one pair of the same path $\vec{y}$ must be $\vec{y}_1$, and the probability of that is $1- C_M^\mathcal{R}/M^{\mathcal{R}}$, which is $0$ for $\vec{y_2}$. 
In other words, if $A(\vec{y})$ does not contain any path repetition, there is a higher probability that $\vec{y}$ is $\vec{y_2}$, otherwise it must be $\vec{y_1}$. 
This issue can be avoided by always accessing a path that has not been accessed after a round of eviction when the required block is stored in the stash. 
In this case, $A(\vec{y})$ never contains repeated paths no matter what  $\vec{y}$ contains, and we just need to ensure $R < M$, which is always the case in practice.
As a result, with the extra dummy access for each read, the adversary cannot distinguish between one sequence and the other.

\end{proof} 


\section{\fname{}}\label{sec:framework}
In this section, we first describe how \fname{} efficiently generates summarized time series data blocks and how these blocks are stored inside \oramname{}. Then, we describe how \fname{} utilizes these blocks to efficiently answer clients' queries.

In \fname{}, a system administrator is responsible for initialising the system by setting the system parameters, including the block size $B$, bucket capacity $Z$, block generation interval $T$, and eviction frequency $R$. Subsequently, \fname{} continuously receives data points from the data producer, generates data blocks, and answers the client's queries. 

\subsection{Block Generation}\label{sec:data_model}
\fname{} stores the TSD as summarised data blocks based on the supported aggregate functions instead of the raw data points. 
The block summarises raw data points of a pre-defined time interval i.e., $\left[t_{i}, t_{i+1}\right)$ with a fixed interval $T=t_{i+1}-t_{i}$. 
Larger intervals provide lower query accuracy but high performance with less storage. To support multiple accuracy levels and higher query performance, \fname{} generates blocks at different time intervals, i.e, aggregation intervals $V$, where $V=[T_1,T_2,....]$ (See \cref{sec:query_optimisation}). The generated data blocks are stored in \oramname{}. Each block contains the aggregated values for the supported aggregate functions for $\left[t_{i}, t_{i+1}\right)$. 
A block is represented by an array, where each item in the array contains all the aggregated values. By storing summarised blocks, \fname{} reduces the ORAM tree size and the query latency. 

\subsection{Query Realisation}\label{sec:query_realisation}
\fname{} receives encrypted queries from clients in the form $Q= E_k(\langle f,(t_a,t_b) \rangle)$ for a range query and $Q=E_k(\langle f,t_a \rangle)$ for a point query, where $f$ is the aggregate function to be executed over the time interval from {$t_a$ to $t_b$}. In the following subsections, we describe how \fname{} handles point and range queries and how they can be extended to support complex analytics.

\noindent \textbf{Point Queries}\label{sec:point_queries}: 
When \fname{} receives a query, it first decrypts the query, i.e., $Q' = D_{k}(Q)$ where $k$ is the client's private key. \fname{} then extracts $t_a$ from $Q'$ and retrieves the data block using \oramname{}. Once the block is retrieved, \fname{} uses the aggregates position map to find the location of the requested $f$'s value. For example, assume that $f = AVG$, \fname{} will access the $avg$ value at index 0 (assuming that $avg$ is located at index 0) of the retrieved block. Finally, the results are encrypted using the client's private key and sent back to the client.

\noindent \textbf{Range Queries}\label{sec:range_queries}
\fname{} follows the same approach that is used with point queries to answer range queries. However, range queries involve retrieving multiple blocks and multiple aggregated values, instead of one. The retrieved values must be fed into the aggregation function to return the final result to the client. Notice that for functions such as min or max, using the retrieved values to calculate the final results is a straightforward process. The reason is that the maximum value for the aggregated values (summarised values) represents the maximum value for the underlying data points. However, this is not valid for other functions, such as the average and variance. To find the average for multiple aggregated values, we use the weighted average, which is given by the formula,
\begin{equation*}
    avg=\frac{\sum_{i=1}^{n} bw_{i} \cdot bx_{i}}{\sum_{i=1}^{n} bw_{i}}
\end{equation*}
where $bw$ is the number of data points in the block (count), $bx$ is the aggregated average value for the block, and $n$ is the number of blocks retrieved by the query. 

The above formula uses only the count and average values from each retrieved block and performs the calculations in a single pass. \fname{} adopts similar approaches for all supported aggregation functions to answer range queries to achieve high performance without compromising accuracy.

\noindent \textbf{Complex Analytics}\label{sec:complex_analytics}
\fname{} can also combine several aggregate values to answer complex range queries. This is possible because \fname{} retrieves all blocks within the queried time interval $(t_a,t_b)$. Unlike cryptographic approaches, blocks’ values are stored in plaintext inside the enclave. Hence, one can easily combine and perform arithmetic operations on the aggregate values.
\fname{} can also in principle support sketch algorithms such as Count-Min, Bloom filter and HyperLogLog. A single- or multi-dimensional sketch table can be flattened and represented by a one-dimensional array. This allows \fname{} to store a sketch table inside the data block as a range of values (instead of a single aggregate value).

\noindent \textbf{Query Optimisation}\label{sec:query_optimisation}
In  \oramname{}, query latency increases linearly with the number of accessed blocks in range queries. The reason is that each block access in \oramname{} is independent of the preceding and subsequent access. Such an approach allows \oramname{} to offer a stronger leakage profile but degrades query performance. To prevent this performance drawback, \oramname{} optimises queries by reducing the number of accessed blocks. \oramname{} achieves this by maintaining multiple ORAM trees with different time intervals $V$ = $[T_0, T_1, T_2,...]$, where $T_{i-1} < T_i < T_{i+1}$, with $T_i \in V$.

Query optimiser works by examining the client's query $Q=\langle SUM,(t_{1},t_{6}) \rangle$ and determining the optimal combination of aggregation intervals to minimise the total number of accessed blocks.
For example, assume that $T_1=60s$ and $T_2=10s$, hence, every block in the $T_1$ tree summarises 6 blocks of $T_2$ tree. Assume that $Q=\langle SUM,(0,70) \rangle$.
Processing this query without optimisation requires accessing 7 blocks in $T_2$ tree, while with query optimisation, \fname{} will access 1 block from each tree, totalling 2 blocks.
As each aggregation interval requires a separate ORAM tree; therefore, $T_i$ will determine the size of the tree. Note that a larger $T_i$ offers high query performance with lower accuracy and less storage. Therefore, such a parameter can be optimised as per the application's requirements.

\subsection{Security of \fname{}}\label{sec:framework_security_guarantees}
\fname{} handles input parameters and queries securely inside the enclave, i.e., parameters and queries are encrypted using the clients' key and can only be decrypted by the enclave. In addition, \fname{} adopts and integrates \oramname{} to obliviously store and access time series data inside the enclave. Therefore, the security definition of \fname{} is captured in the following theorem.
\begin{theorem}
\fname{} \cref{sec:framework} is secure giving that \oramname{} \cref{sec:proposed_oram} is proven to be secure in Theorem \ref{def:oram}.
\end{theorem}

\begin{proof} Giving that \fname{} handles data security inside the enclave, therefore, the adversary cannot examine the system parameters or client queries in plaintext. 
In detail, the adversary cannot examine the following: block size, bucket size $Z$, the total number of blocks $N$, eviction frequency $\mathcal{R}$ and query function. Although the adversary can observe which tree is accessed during quqry optimisation \cref{sec:query_optimisation}, this will only reveal $T$ but not the exact query range (i.e., $(t_a,t_b)$). The reason that \fname{} does not protect the block generation time-interval $T$ is that the adversary can easily observe such an access pattern during block generation.
In addition, \fname{} does not protect the tree, stash, path lookup, or stash lookup sizes, as the adversary can monitor the memory pages and infer them.
Finally, reading and writing blocks from/into \oramname{} is oblivious and secure, as discussed in \cref{sec:oram_security_guarantees}. 
Furthermore, all other operations performed by \fname{}, i.e., block generation and the access to the of the required aggregated values from blocks accessed, are also oblivious as they are performed with oblivious primitives \cref{sec:oblivious_primitives}. Therefore, \fname{} does not leak any memory access pattern to the server. 

\end{proof} 


\section{Evaluation}\label{sec:evaluation}
In this section, we evaluate \fname{} while asking the following questions:
\begin{enumerate}
    \item What is the performance of \fname{} compared to the non-oblivious version and ORAM Baseline?
    \item How do the internal components of \oramname{} and the levels of aggregation affect its performance?
\end{enumerate}

\noindent\textbf{Implementation}. We implemented \fname{} and \oramname{} in $\sim$4,000 lines of code. Inside the enclave (server-side), we use the trusted cryptography library provided by SGX SDK, i.e., \textit{sgx\_tcrypto} \cite{intel2020sdk}. As mentioned earlier in Section~\cref{sec:supported_aggregates}, we implemented different aggregate functions 
where the value of each function is represented by a single item in the data block. We store values in an array of floating-point numbers (4-bytes per number) in the data blocks.

\noindent\textbf{Expiremental Setup}
We evaluated \fname{} on a local network. For the server, we use an SGX-enabled server running on Intel \textregistered Xeon CPU \textregistered E-2288G @ 3.70 GHz with 8 cores and 128 GiB RAM. The enclave size is 256 MB. We simulate a client with 1 vCPU and 3.75 GB memory. We do not include the network latency in our evaluation as our goal is to measure the performance of \fname{} on the most critical part, i.e., the server-side.

\noindent{\bf Dataset}. 
We use the time series benchmark suite \cite{timecryptBenchmark} to generate CPU utilization dataset. The dataset contains a single attribute (i.e., CPU usage). We initialise \fname{} to store 24 hours of readings (i.e., CPU usage), where each data block in the ORAM tree represents 10s of readings (i.e., $T=10s$). Each block stores 10 aggregate values, each value consumes 4-bytes ($B=40$ bytes). Each bucket in the tree stores 4 blocks ($Z=4$). Therefore, the height of each tree in \oramname{} is $L=13$. 
 
\subsection{Baselines}
\noindent{\bf ORAM baseline}.
In this paper, we introduce and integrate our \oramname{} into \fname{} to achieve obliviousness with minimal query blocking inside the enclave. We evaluated \oramname{} against the widely adopted ORAM, i.e., PathORAM \cite{stefanov2018path}. We integrate both PathORAM and \oramname{} into \fname{} for evaluation. For a fair comparison, as done in \oramname{}, we stored the stash, position map, and the tree in plaintext inside the enclave for PathORAM. Moreover, we store 4 blocks in each bucket ($Z=4$) for both ORAMs.
However, PathORAM performs the eviction after each read/write access.

\noindent{\bf Non-oblivious baseline}.  To understand the overhead of oblivious operations, we evaluate \fname{} without \oramname{} and oblivious operations (referred to as non-oblivious). In detail, \fname{} generates data blocks and stores them inside the enclave in a non-oblivious data structure. Thus, we replace \oramname{} with non-oblivious storage. For the position map, \fname{}  stores block IDs in a key-value structure with a pointer to the block's physical address. To answer clients' queries, \fname{} retrieves the required block id from the position map and accesses (and aggregates) data blocks using non-oblivious operations. Such a setup allows us to evaluate the overhead of obliviousness in \fname{}, i.e., \oramname{} and oblivious operations. 

\subsection{Evaluation Results}

\subsubsection{\textbf{Query latency}}\label{sec:query_latency}
To understand \fname{}'s performance, Table \ref{tbl:query_latency} shows the query latency for different query ranges and block sizes. We evaluated \fname{} using reasonable small block sizes, since the size represents the number of supported aggregates, where a block size of $2^7$ can store up to 32 aggregate values. As expected, the query latency increases linearly with the query range. The reason is that for each queried time interval, \fname{} needs to access one path in \oramname{} and retrieves $(L+1) \cdot Z$ blocks from the ORAM. Similarly, the larger the block size is, the more overhead it will add to the query performance, which is expected in a tree-based ORAM. 

\begin{table}
    \vspace{-20pt}
    \renewcommand{\arraystretch}{1.3}
    \scriptsize
    \centering
        \caption{\fname{} Query latency (ms) for point and different range queries and B block sizes. Note that $R_T$ is a range of intervals, where each interval represents a single aggregated data block.
        }
        \vspace{0.2cm}
    \begin{tabularx}{1\textwidth}{p{2cm} | p{2cm} p{2cm} p{2cm} p{2cm} p{2cm}}
        \hline 
        Range ($R_T$) & \multicolumn{5}{c}{-- \textit{Block size (bytes)} --}
        \\
         & $B=2^3$ & $B=2^4$ & $B=2^5$ & $B=2^6$ & $B=2^7$  \\
        \hline 
        1 & 0.03 & 0.03 & 0.03 & 0.03 & 0.04
        \\
        10 & 0.12 & 0.12 & 0.12 & 0.13 & 0.15
        \\
        20 & 0.19 & 0.2 & 0.21 & 0.22 & 0.26
        \\
        30 & 0.28 & 0.28 & 0.29 & 0.31 & 0.38
        \\
        40 & 0.37 & 0.36 & 0.38 & 0.41 & 0.49
        \\
        50 & 0.46 & 0.46 & 0.46 & 0.5 & 0.59
        \\
        \hline
    \end{tabularx}
    \vspace{-10pt}
    \label{tbl:query_latency}
\end{table}

Fig. \ref{fig:query_latency_breakdown} shows the query latency breakdown with a block size $B=2^7$. The majority of the overhead is due to ORAM and oblivious operations for point and range queries. By omitting the SGX overhead (as it is consistent with different query ranges), the ORAM operations overhead comprises between 80-85\% of the query latency. On the other hand, the overhead of oblivious operations comprises between 15-20\% of the query latency. Note that the ORAM overhead is reduced in \fname{} by query optimisation as shown in Section~\cref{eval:query_optimisation}.

\begin{figure}
\vspace{-10pt}
  \centering
  \hfil
  \begin{subfigure}{0.49\textwidth}
    \centering
    \includegraphics[width=1\linewidth]{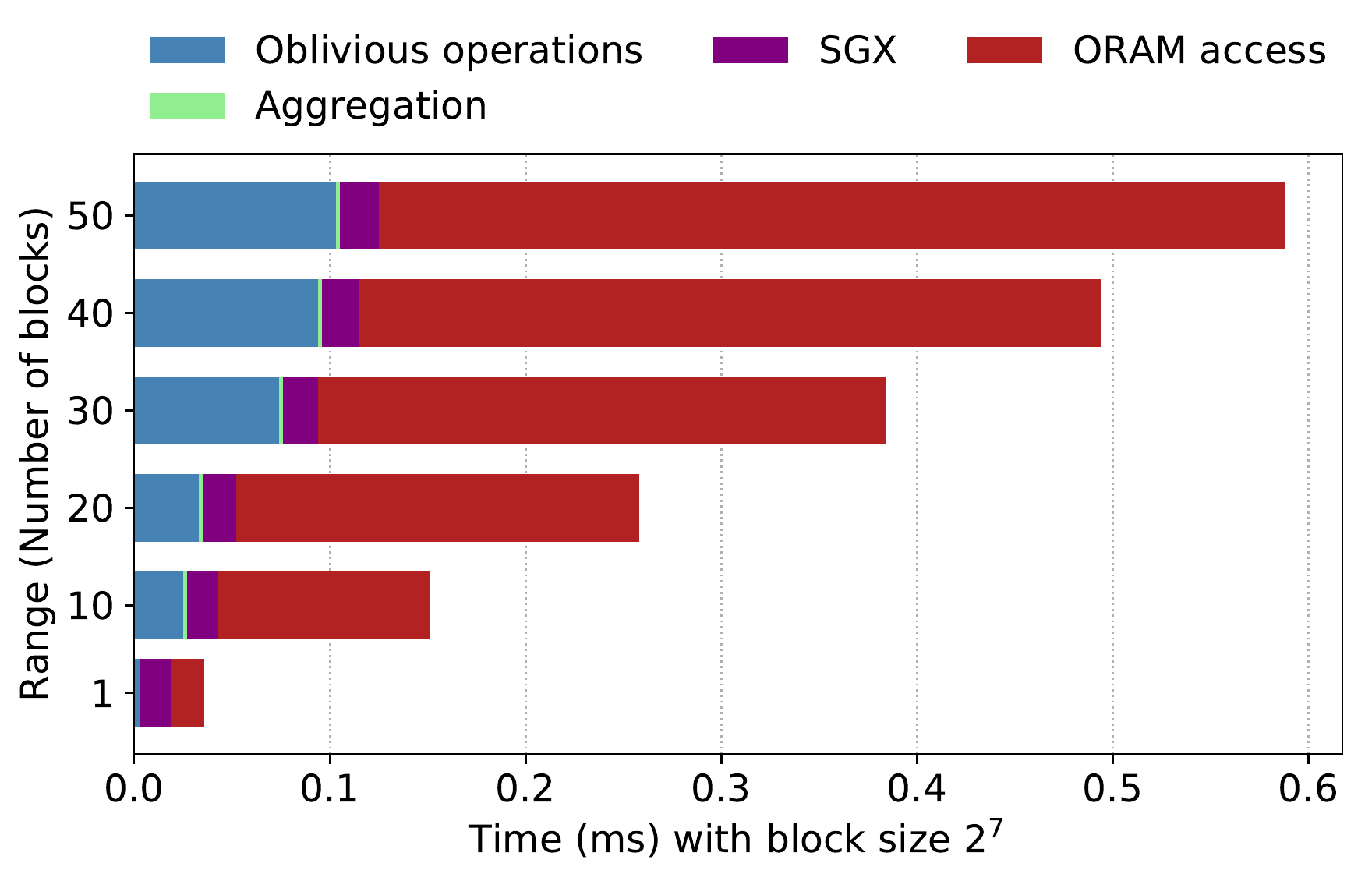}
    \caption{Latency breakdown}
    \label{fig:query_latency_breakdown}
  \end{subfigure}
  \hfil
  \begin{subfigure}{0.49\textwidth}
    \centering
    \includegraphics[width=1\linewidth]{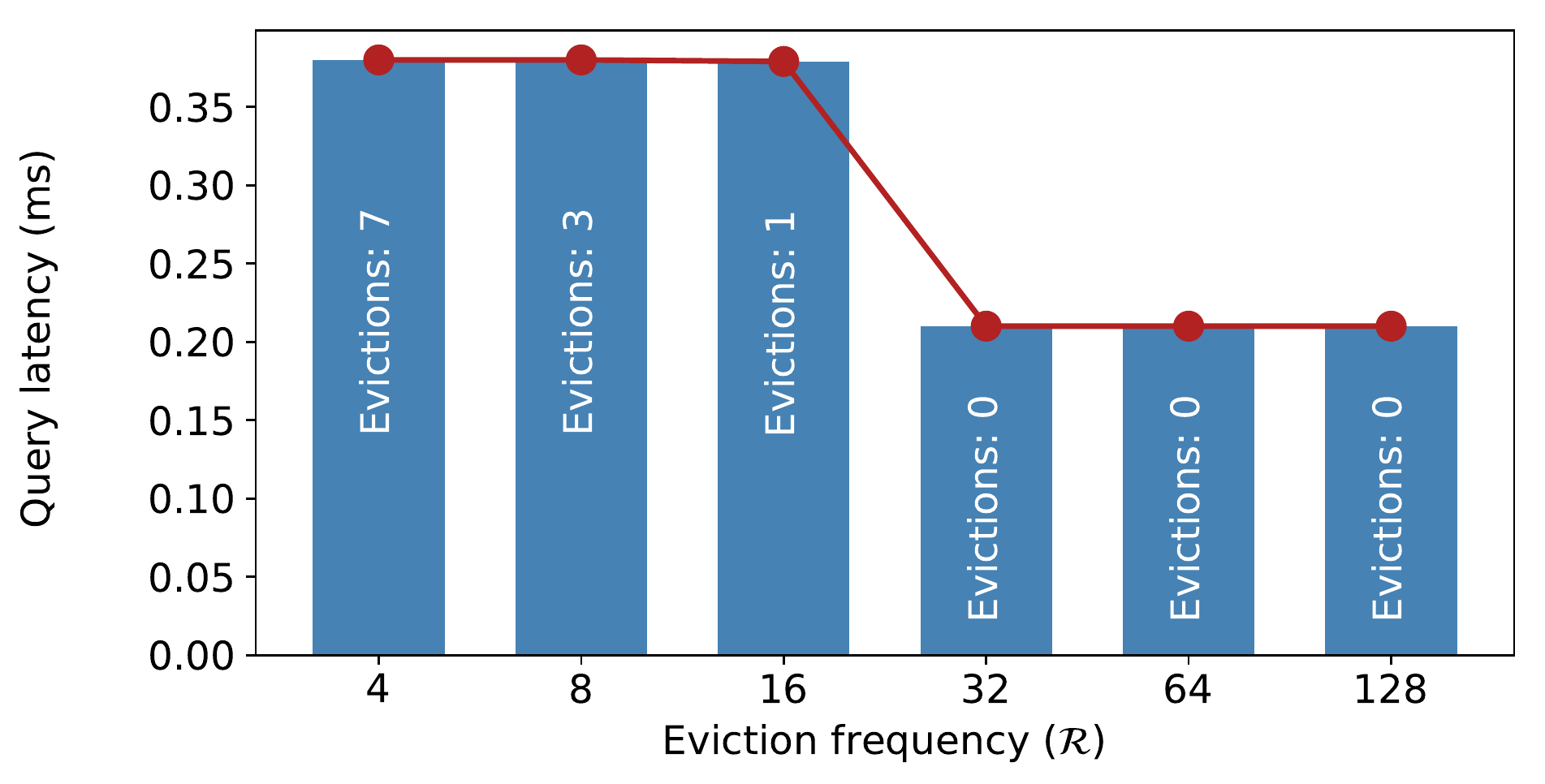}
    \caption{Query latency}
    \label{fig:eviction_frequency}
  \end{subfigure}
  \caption{A) Query latency breakdown with different query ranges ($R_T$). SGX overhead includes context switch and memory allocation inside the enclave. B) Query latency with different eviction frequencies ($\mathcal{R}$) and query range = 32 intervals.}
  \label{fig:results_query_latency}
  \vspace{-15pt}
\end{figure}

\noindent\textbf{Eviction frequency}. 
\fname{} evicts the blocks from the stash for every $\mathcal{R}$ read-operation. Fig. \ref{fig:eviction_frequency} demonstrates how $\mathcal{R}$ affects the query latency. With a fixed query range of 32, the query has a negligible higher latency when the value of $\mathcal{R}$ is smaller than the query range, i.e, $t_b - t_a \leq \mathcal{R}$. However, the latency drops significantly when $\mathcal{R}$ is larger than the query range. This is because \oramname{} performs read-operations from the ORAM with less evictions.

\noindent\textbf{Aggregation Intervals}
\label{eval:query_optimisation}.
In Fig.~\ref{fig:aggregation_intervals}, we show how aggregation levels reduce query latency in \fname{}. We set $T=1$s as the baseline in this evaluation, i.e., \fname{} generates 1 block per second. Noticeably, query latency decreases with larger aggregation intervals as less accesses occur to the ORAM. When $T=20$s, the speedup achieved in the query latency is $67\times$ with $20\times$ less memory consumption (compared to $T=1$s). Note that \fname{} supports multiple aggregation levels by maintaining a separate ORAM tree for each level. Despite the fact that \fname{} consumes less memory for higher aggregation levels, such an overhead can be neglected due to the large EPC size in SGX v2.

\begin{figure}
    \vspace{-10pt}
\centering
    \includegraphics[width=0.7\linewidth]{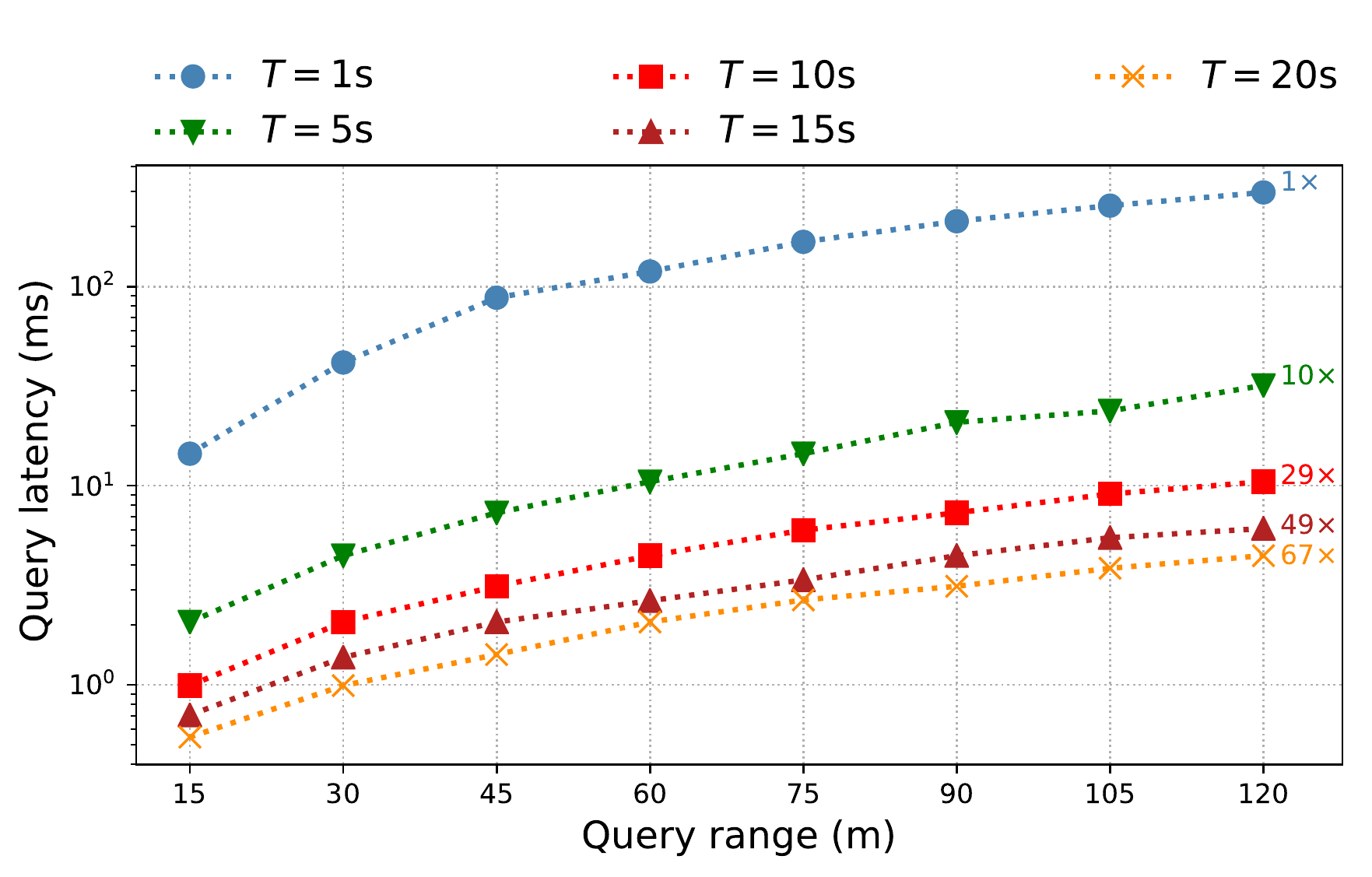}
    \caption{Query latency of different aggregation intervals.}
    \label{fig:aggregation_intervals}
    \vspace{-20pt}
\end{figure}

\noindent \textbf{Comparison with baselines}. 
Fig. \ref{fig:baseline_comparison} illustrates how \fname{}'s query latency compares to the baselines. The latency grows linearly with the query range for \fname{} and the baselines. For point queries, the latency overhead of \fname{} and ORAM baseline compared to the non-oblivious \fname{} is $1.5\times$ and $2.5\times$ respectively. \oramname{} achieves higher performance than the ORAM baseline due to its non-blocking read operations and efficient batch eviction design. For range queries, \fname{} and ORAM baseline adds up to $12\times$ overhead to the query latency compared to non-oblivious \fname{}. As discussed earlier, the majority of the overhead is due to the ORAM access and the oblivious operations in \fname{} and ORAM baseline. However, \fname{} remains substantially faster than the ORAM baseline by $1.7-2\times$ for both range and point queries.

\begin{figure}
\vspace{-10pt}
  \centering
  \hfil
  \begin{subfigure}{0.49\textwidth}
    \centering
    \includegraphics[width=1\linewidth]{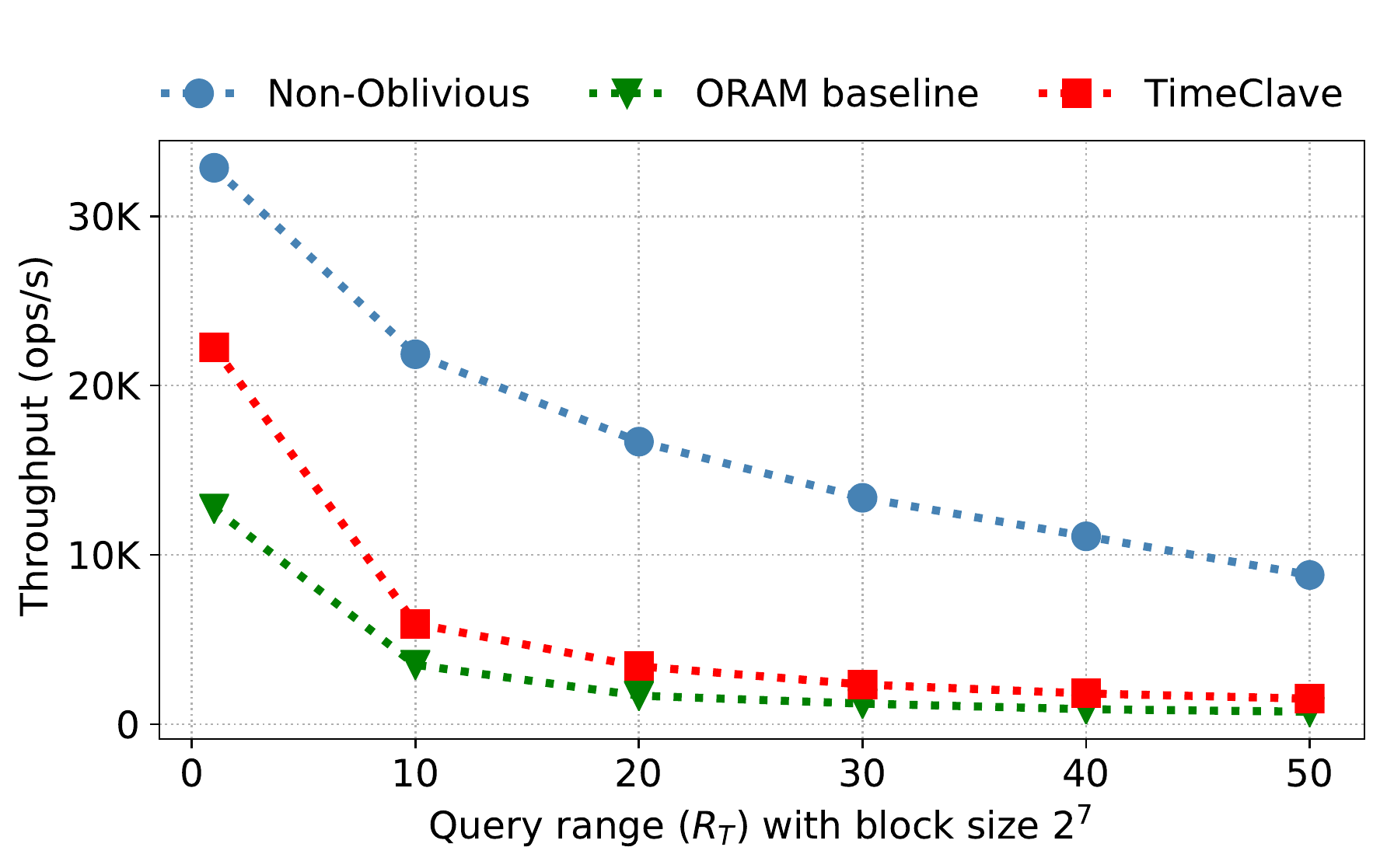}
    \caption{Query throughput}
    \label{fig:query_throughput}
  \end{subfigure}
  \hfil
  \begin{subfigure}{0.49\textwidth}
    \centering
    \includegraphics[width=1\linewidth]{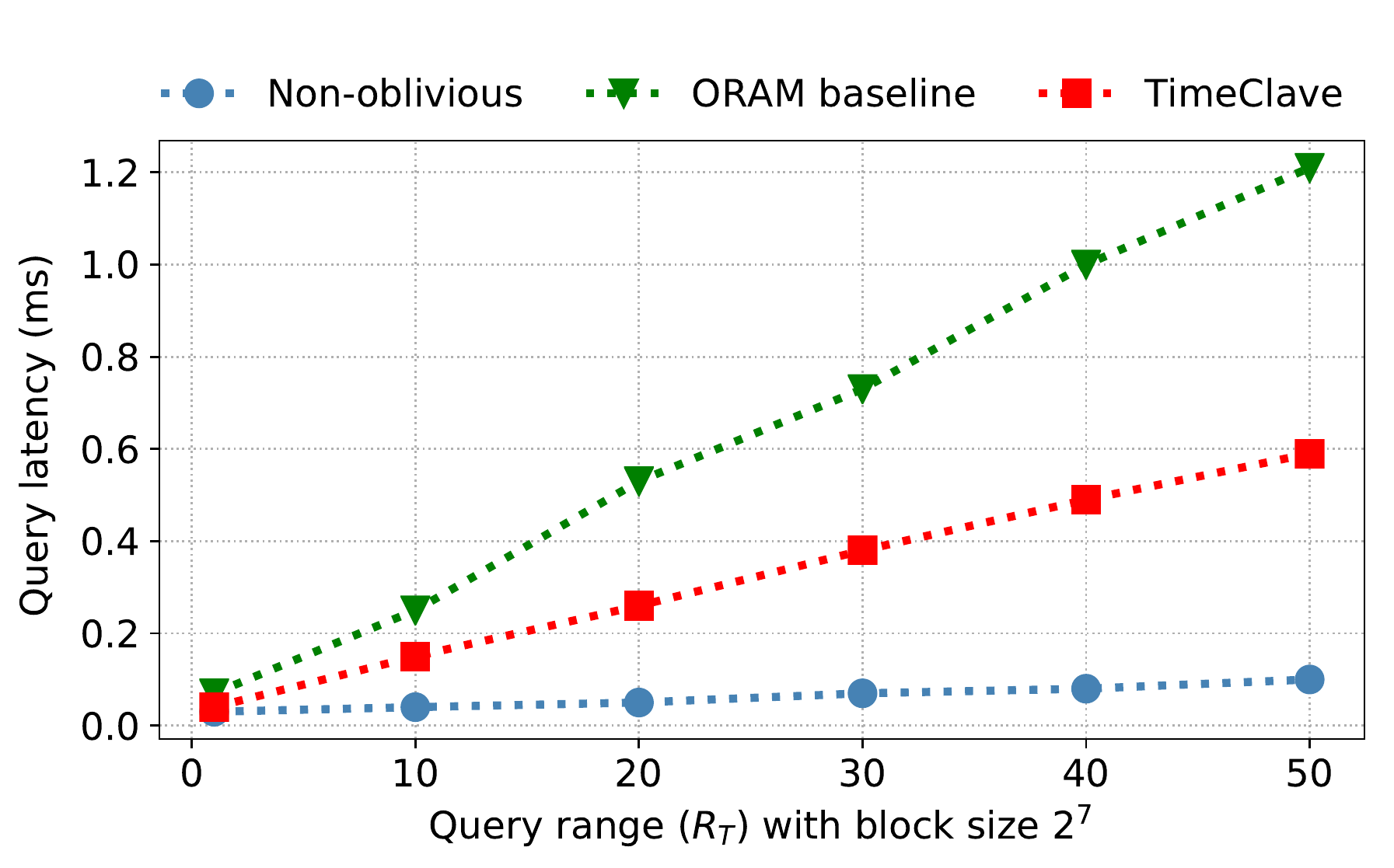}
    \caption{Query latency}
    \label{fig:baseline_comparison}
  \end{subfigure}
  \caption{\fname{} query latency and throughput compared to baselines.}
  \label{fig:results_query_throughput}
  \vspace{-20pt}
\end{figure}

\subsubsection{\textbf{Query Throughput}}

In Fig. \ref{fig:query_throughput}, we compare \fname{} query throughput with the baselines. For point queries, \fname{} achieves up to 22K ops/s compared to the ORAM baselines which achieves up to 12K ops/s. For range queries, \fname{} achieves higher throughput than ORAM baseline by up to $2\times$, i.e., 1.5K ops/s for $R_T \geq  20$. Similar to the query latency both \fname{} and ORAM baseline add up to $6\times$ overhead compared to non-oblivious \fname{}. As mentioned in Section \cref{sec:query_latency}, the majority of the overhead is caused by ORAM access and oblivious operations. Such overhead can be reduced in \fname{} by maintaining multiple aggregation intervals, which reduces the number of accessed paths in \oramname{}. For this, we avoid large query ranges in our evaluation as the main goal of \fname{} is to summarise TSD and maintain multiple aggregation intervals to achieve low query latency.

\subsubsection{\textbf{Memory Consumption}}
\begin{table}
    \renewcommand{\arraystretch}{1.3}
    \scriptsize
    \vspace{-35pt}
    \centering
        \caption{Memory consumption using different aggregation intervals with the expected query latency speedup when the same time interval is queried.}
        \vspace{0.2cm}
        \label{tbl:memory_consumption}
    \begin{tabularx}{1\textwidth}{p{3cm} | p{4cm} p{4cm}}
        \hline 
        $V$ & Memory (MB) & Speedup
        \\
        \hline 
        $[10s]$ & 5.0 & up to $1\times$
        \\
        $[10s, 60s]$ & 5.6 & up to $6\times$
        \\
        $[10s, 120s]$ & 5.3 & up to $12\times$
        \\
        \hline
\end{tabularx}
\end{table}

\oramname{} uses two separate ORAM trees to handle non-blocking queries. As a result, it consumes twice the data as the ORAM baseline. In detail, \fname{} maintains two separate ORAM trees to store the data blocks, two stashes, two position maps, and the list of accessed paths. 
For example, assume $T=10s$, $\mathcal{R}=4$ and $B=40$ bytes.
 To store 24h of time-series data, \fname{} will generate $N=8,640$ blocks. Recall from \cref{sec:storage}, \oramname{} initialises two ORAM trees, each tree requires $2^{L+1} -1 \cdot Z \cdot B$ bytes., i.e., $\sim$5MB for both trees to store a single attribute (e.g., CPU, memory). To maintain multiple aggregation intervals, e.g., $V=[10s, 60s]$, \oramname{} will consume $\sim$5.6MB while achieving up to $6x$ query speedup with only $0.12\times$ memory overhead. Although \fname{} offers a low memory overhead, \fname{} utilises the new generation of SGX-enabled processors support up to 1TB of enclave memory \cite{intelsgx1tb}.

\subsubsection{\textbf{\fname{} Compared to Cryptographic Approaches}}
We evaluate \fname{} against the cryptographic solutions TimeCrypt \cite{burkhalter2020timecrypt}, which supports aggregate functions over encrypted time-series data.
Without query optimization, \fname{} achieves up to $16x$ lower query latency when the number of queried blocks is below 6,000. TimeCrypt exhibits an almost constant latency of around 185$ms$. However, with query optimization using multiple aggregations intervals, \fname{} demonstrates a significant improvement in performance in comparison to TimeCrypt by orders of magnitude. Although the query latency for \fname{} increases with the number of queried blocks, querying 8,000 blocks results in a query latency that is approximately $200x$ lower than TimeCrypt.

Without query optimisation, \fname{} achieves better performance when the number of queried blocks is smaller than 6,000, while TimeCrypt shows almost a constant latency of around 185$ms$. Nevertheless, it is clear that \fname{} can achieve up to $16x$ better performance than TimeCrypt without the query optimiser when the number for small range queries. With query optimisation, \fname{} achieves better performance than TimeCrypt by orders of magnitude. Although the query latency for \fname{} increases with the number of queried blocks, querying 8,000 blocks has around $200x$ lower query latency.


\section{Related work}\label{sec:related_work}

\subsection{Secure Time Series Processing}
Cryptographic protocols have been widely adopted in building secure databases \cite{demertzis2020seal, papadimitriou2016big, poddar2016arx} to execute expressive queries on encrypted data, while another line of work leverages SGX, such as Oblidb \cite{eskandarian2017oblidb}, EncDBDB \cite{fuhry2021encdbdb}, EnclaveDB \cite{priebe2018enclavedb} and Oblix \cite{mishra2018oblix}. However, these solutions either incur significant performance overhead or are not optimised for time series processing. The most related works to \fname{} in the secure time series processing systems are TimeCrypt \cite{burkhalter2020timecrypt}, Zeph and Waldo \cite{dauterman2021waldo}. TimeCrypt and Zeph employ additive homomorphic encryption to support aggregated queries on encrypted data. However, both solutions are non-oblivious, allowing the adversary to learn sensitive information by recovering search queries or a portion of the encrypted records \cite{zhang2016all,liu2014search}. On the other hand, Waldo \cite{dauterman2021waldo} offers a stronger security guarantee than \cite{burkhalter2020timecrypt,dauterman2021waldo} by hiding query access patterns. As Waldo adopts MPC, the network bandwidth adds significant overhead to the query latency. \fname{} eliminates such overhead while providing fully-oblivious query processing.
Unlike previous solutions, \fname{} can be easily scaled to support complex analytics as it processes TSD in plaintext inside the enclave.

\subsection{ORAM with Intel SGX}
Another line of prior work has explored and combined SGX and ORAM to build secure storage.
For example, ZeroTrace \cite{sasy2017zerotrace} develops a generic oblivious block-level ORAM controller inside the enclave that supports multiple ORAMs.
Additionally, ZeroTrace focuses on hiding memory access patterns inside the enclave while leaving ORAM storage outside the enclave. Similarly, Oblix \cite{mishra2018oblix} builds an oblivious search index for encrypted data by using SGX and Path ORAM. Oblix designs a novel data structure (ORAM controller) to hide access patterns inside the enclave.
Likewise, Oblix stores the ORAM storage on the server in unprotected memory (outside the enclave).
Moreover, Obliviate \cite{ahmad2018obliviate} and POSUP \cite{hoang2019hardware } adopt SGX and ORAM to develop a secure and oblivious filesystem to read and write data from a file within an enclave. Obliviate is optimised for ORAM write operations by parallelising the write-back process to improve performance.
MOSE \cite{hoang2020mose} adopts Circuit-ORAM for a multi-user oblivious storage system with access control. Like previous solutions, MOSE stores the ORAM controller inside the enclave while leaving the ORAM tree outside. Although MOSE parallelises the ORAM read process, clients' queries are blocked until the accessed blocks are evicted.
Unlike \fname{}, previous solutions are not optimised for handling multi-user, non-blocking queries in the time series context.

\subsection{Plaintext Time Series Processing}
Many solutions  have been proposed to store and process TSD in the plaintext domain, to name a few, such as Amazon TimeStream \cite{amazontimestream}, InfluxDB \cite{influxdb}, Monarch \cite{adams2020monarch},  Gorilla \cite{pelkonen2015gorilla}. These solutions focus mainly on delivering low-latency real-time queries and efficient storage. For example, Gorilla \cite{pelkonen2015gorilla} is an in-memory database for TSD that is optimised for heavy read and write operations. It introduces a new compression scheme for timestamps and floating point values for efficient TSD storage and achieves up to $70\times$ lower query latency than on-disk storage. Monarch \cite{adams2020monarch} is used to monitor and process time series data in a geo-distributed architecture. In addition, it supports in-memory data aggregation for higher query performance. Unfortunately, adopting these solutions for sensitive time series systems on cloud platforms can lead to critical data breaches \cite{information2020world}.


\section{Conclusion and Future Work}
In this work, we presented \fname{}, a secure in-enclave time series processing system. 
While previous works \cite{dauterman2021waldo, burkhalter2020timecrypt} adopt cryptographic protocols, \fname{} leverages Intel SGX to store and process ttime series data efficiently inside the enclave.
To hide the access pattern inside the enclave, we introduce an in-enclave read-optimised ORAM named \oramname{} capable of handling non-blocking client queries. 
\oramname{} decouples the eviction process from the read/write operations.
\fname{} achieves a lower query latency of up to $2.5\times$ compared to our ORAM baseline and up to  $5.7-12\times$ lower query latency than previous works.

While \fname{} supports a wide range of aggregate functions that are supported by time series processing systems \cite{amazontimestream,influxdb, timescale}, it does not support all the functionalities (e.g., filtering, grouping, and joining). Note that \fname{} can be easily extended to support such functionalities and expressive queries. This is because \fname{} stores and processes data inside the enclave in plaintext, allowing it to perform complex operations efficiently. However, such a direction requires careful implementation and optimisation to maintain \fname{}'s obliviousness inside the enclave. Finally, \fname{} can further improve its query performance by parallelising each ORAM read access, which can be a valuable direction for future work.

\section*{Acknowledgement}

The authors would like to thank the anonymous reviewers for their valuable comments and constructive suggestions. The work was supported in part by the ARC Discovery Project (DP200103308) and the ARC Linkage Project (LP180101062).

\bibliographystyle{IEEEtran}
\bibliography{main}

\end{document}